\documentclass[journal]{IEEEtran}

\usepackage{xcolor}
\usepackage{graphicx}
\usepackage{bm}
\usepackage{textcomp, gensymb}
\usepackage{amsthm}
\usepackage{amssymb}
\usepackage{amsmath}
\usepackage{enumitem}
\usepackage{subcaption}
\usepackage[bbgreekl]{mathbbol}
\usepackage[ruled,linesnumbered]{algorithm2e}
\usepackage{url}

\usepackage{hyperref}
\hypersetup{
    colorlinks=true,
    linkcolor=blue,
    filecolor=blue,   
    citecolor=blue,   
    urlcolor=blue,
    pdftitle={Overleaf Example},
    pdfpagemode=FullScreen,
    }
\usepackage{cleveref}
\crefname{algocf}{alg.}{algs.}
\Crefname{algocf}{Algorithm}{Algorithms}

\newtheorem{theorem}{Theorem}
\newtheorem{lemma}[]{Lemma}
\newtheorem{conjecture}[]{Conjecture}

\newtheorem{definition}[]{Definition}

\crefname{conjecture}{conjecture}{conjectures}
\Crefname{conjecture}{Conjecture}{Conjectures}

\newcommand{\blambda}{\bm{\lambda}}

\newcommand{\tblambda}{\tilde{\bm{\lambda}}}
\newcommand{\tbblambda}{\tilde{\mathbb{\bblambda}}}

\newcommand{\sblambda}{\blambda^\star}
\newcommand{\bLambda}{\bm{\Lambda}}

\newcommand{\loss}{{\ell_{\bm{p}}}}
\newcommand{\ploss}{\ell_{\bm{p}'}}
\newcommand{\bloss}{\bm{\ell}_{\bm{p}}}

\allowdisplaybreaks

\begin{document}
%
% paper title
% Titles are generally capitalized except for words such as a, an, and, as,
% at, but, by, for, in, nor, of, on, or, the, to and up, which are usually
% not capitalized unless they are the first or last word of the title.
% Linebreaks \\ can be used within to get better formatting as desired.
% Do not put math or special symbols in the title.
\title{Solar Photovoltaic Systems Metadata Inference and Differentially Private Publication}

% author names and affiliations
% transmag papers use the long conference author name format.

\author{
Nikhil~Ravi,~\IEEEmembership{Graduate Student Member,~IEEE},
Anna~Scaglione,~\IEEEmembership{Fellow,~IEEE},
Julieta Giraldez,~\IEEEmembership{Member,~IEEE},
Parth Pradhan,~\IEEEmembership{Senior Member,~IEEE},
Chuck Moran,
Sean~Peisert,~\IEEEmembership{Senior Member,~IEEE}% <-this % stops a space
\thanks{N. Ravi and A. Scaglione are with the Department of Electrical and Computer Engineering, Cornell Tech, e-mail: nr337@cornell.edu.
J. Giraldez, P. Pradhan, and Chuck Moran are with Kevala.
S. Peisert is with Lawrence Berkeley National Laboratory.
}
\thanks{
The Director, Cybersecurity, Energy Security, and Emergency Response, Cybersecurity for Energy Delivery Systems program, of the U.S. Department of Energy, under contract DE-AC02-05CH11231 supported this research.  Any opinions, findings, conclusions, or recommendations expressed in this material are those of the authors and do not necessarily reflect those of the sponsors of this work.}%
}% <-this % stops a space

\IEEEtitleabstractindextext{%
	\begin{abstract}
		Stakeholders in electricity delivery infrastructure are amassing data about their system demand, use, and operations. Still, they are reluctant to share them, as even sharing aggregated or anonymized electric grid data risks the disclosure of sensitive information. This paper highlights how applying differential privacy to distributed energy resource production data can preserve the usefulness of that data for operations, planning, and research purposes without violating privacy constraints. Differentially private mechanisms can be optimized for queries of interest in the energy sector, with provable privacy and accuracy trade-offs, and can help design differentially private databases for further analysis and research. In this paper, we consider the problem of inference and publication of solar photovoltaic systems' metadata. Metadata such as nameplate capacity, surface azimuth and surface tilt may reveal personally identifiable information regarding the installation behind-the-meter. We describe a methodology to infer the metadata and propose a mechanism based on Bayesian optimization to publish the inferred metadata in a differentially private manner. The proposed mechanism is numerically validated using real-world solar power generation data.
	\end{abstract}

	% Note that keywords are not normally used for peerreview papers.
	}

\maketitle
\IEEEdisplaynontitleabstractindextext
\IEEEpeerreviewmaketitle

\section{Introduction}
\IEEEPARstart{U}{tilities} and distributed energy resource (DER) providers are significantly expanding their data collection capabilities and are undergoing data-driven, digital transformations of their organizations~\cite{currie2023data}. Concurrently, there is also significant interest in accessing this data by third parties; for example:
\begin{itemize}[leftmargin=*]
	\item Organizations planning optimal placements for DER sites.
	\item Commercial industries seeking breakthroughs for the digital economy in the energy sector.
	\item Law enforcement agencies investigating cyber-crimes on these infrastructures.
    \item Regulatory bodies working to encourage decarbonization and grid modernization goals \cite{unitedstatesdepartmentofenergyofficeofelectricity2020strategy}.
\end{itemize}

However, the growing number of measurements collected in distribution systems has exacerbated data privacy issues~\cite{liu2012cyber}. The data can reveal user behavior, expose retail market vulnerabilities (such as the times of congestion in the system), or expose company proprietary information. Simultaneously, the analysis of this valuable data enables necessary transformations in the energy sector, such as developing technologies for responsive loads in buildings and electrified transportation or highlighting market opportunities to spur investments in deploying renewable energy. Furthermore, the tie-lines of the physical grid allow failures to spread across management boundaries; thus, sharing information across such management zones and with government and regulatory agencies could help develop stronger defensive postures, raise alerts about imminent or ongoing attacks, and support recovery from successful attacks.

The prevalent vision of the grid includes a trusted curator, as seen in \cref{fig:trusted_curator}. The curator may be the electric utility itself or a contractor. Utilities inherently silo data, making sharing very difficult, out of concern that sharing certain data can implicitly violate the privacy of customers. In addition, the default operational security paradigm for the past two decades has been to protect utility grids for national security reasons. While regulators are requiring utilities to share more data, there is still a fear -- among both the consumers and the operators -- of sharing data for privacy reasons~\cite{lamm2021data}.
\begin{figure}[!htbp]
	\centering
	\includegraphics[width=0.45\textwidth]{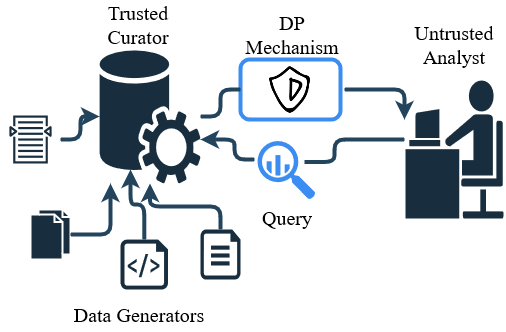}
	\caption{Data collection with a trusted central curator.}
	\label{fig:trusted_curator}
\end{figure}

\subsection{Differential privacy for measured energy data}
The electric utility must abide by strict laws about customer privacy while publishing data or aggregate query answers to external agents. We highlight how applying differential privacy (DP) to DER production data can preserve the usefulness of that data for operations, planning, and research purposes without violating privacy constraints. DP mechanisms can be optimized for queries of interest in the energy sector, with provable privacy and accuracy trade-offs, and can help design differentially private databases for further analysis and research.

DP can be loosely defined as a controlled amount of statistical ``noise'' introduced to data, intended to protect aspects of data that could otherwise be used to re-identify information associated with it. This noise shifts numerical data within a definable threshold such that, while the data no longer truly represents original measurements (making the data harder to trace to sensitive personally identifiable information (PII)), it is still representative enough to be useful for an intended analytical purpose.

Unlike anonymization techniques, DP mechanisms can provide mathematical guarantees on the amount of information about a dataset that can be leaked to a data analyst or any other party who receives it. Using DP in the context of an analyst studying a database allows for the provision of a privacy budget to limit the maximum information leakage over a set of queries sent to that database; the DP outputs will contain approximate statistical answers to queries and analyses that can be optimized for utility and acceptable privacy leakage.

General DP mechanisms are practical and readily available, creating an opportunity to analyze data relevant to cybersecurity for energy delivery systems while preserving privacy. For these systems to capitalize on this opportunity, the application of DP mechanisms to energy datasets must be analyzed. Successfully applying DP mechanisms to energy datasets can lead to increased comfort for energy stakeholders (consumers, utilities, equipment vendors, and government regulators) with the use of their data for various analytical and planning purposes, which could result in significant increases in data sharing while protecting privacy, trade secrets, and other sensitive information in the data.

\subsection{Privacy needs for Solar Photovoltaics' production data and related contributions}\label{sec:SolarPV}
The rapid expansion in the adoption of solar photovoltaics (PV) poses new challenges in the operation of power systems. High solar adoption can disrupt the supply-demand balance on the grid due to the increased need for electricity generators to quickly ramp up energy production when the sun sets and the contribution from PV falls. The unique change in the shape of the electric load met by conventional power plants coupled with increasing levels of PV generation in the system is now commonly referred to as the \textit{duck curve}. The duck curve highlights concerns that the conventional power system will be unable to accommodate the ramp rate and range needed to fully utilize solar energy. Another challenge with high solar adoption is the potential for PV to produce more energy than can be used at one time, called over-generation ~\cite{phuangpornpitak2016study,frew2021curtailment}. This leads system operators to curtail PV generation, reducing its economic and environmental benefits. As such, the ability to predict the net load forecast (load minus solar generation) is becoming crucial to the safe and reliable operation of the grid.

Tracking the performance and statistical forecasting of PV systems' power generation is essential to a utility's ability to perform control operations to ensure the stability of the grid, for PV performance and forecasting activities, and to detect faults to the systems~\cite{meng2020datadriven}. This information is also important for the further penetration of solar PV systems by showcasing the profitability of the adoption of PV systems. The availability of accurate metadata about the various distributed solar PV systems in the field is necessary to perform the aforementioned analyses. This includes information about the PV systems' nameplate capacity, surface azimuth, and surface tilt parameters. However, this is sensitive information, which when revealed could enable various market attacks.

The Federal Energy Regulatory Commission (FERC) approved Order No. 2222 in September 2020~\cite{ferc2020new}. This landmark decision requires the regional transmission organizations and independent system operators that manage the transmission grids providing electricity to about two-thirds of the country to figure out how to give DERs access to wholesale energy markets. It promotes competition in electric markets by removing the barriers preventing DERs from competing on a level playing field in the organized capacity, energy, and ancillary services markets run by regional grid operators. The capability of DERs to participate in markets will also increase the creation and participation of DER aggregators, and the estimation of residential and commercial PV will become of growing importance and prone to cybersecurity and privacy concerns~\cite{americanpublicpowerassociation2018comments}.

The implementation of smart meters has increased concern about consumer privacy because the fine-grained meter data being collected could be used to infer the activities and consumption and production behavior patterns of consumers. These data could be shared between utility companies, third parties and other stakeholders to improve the planning and operations of the grid, as well as the way customers and DERs interact with it. A significant number of actors including DER aggregators, Microgrid operators, OEMs, vendors, and other third parties will be interested in inferring the PV production connected to the primary and secondary side of the distribution feeders. Utilities will only collect the net demand from a smart meter. Larger systems in the hundreds of kilowatts (kW) to megawatt (MW) scale directly connected to the distribution grid (versus BTM solar PV) will more commonly be monitored by utilities via SCADA systems.

On the other hand, some DER data may be inaccessible to researchers and modelers forecasting demand. Consider a rooftop PV system: this system is installed BTM, so granular data may not be reported or available to an electric utility, which typically only has knowledge of the electrical characteristics up to and including the electricity meter on a home (described as front-of-the-meter, or FTM). More than likely, the homeowner will interact with the PV system through a proprietary application developed by a separate company responsible for manufacturing, installing, or maintaining the equipment (nominally the inverter). Typically, the utility is not tied into this monitoring ecosystem (and the production data therein), but aggregating and modeling PV sizing and production data is essential for stable grid operations and proper dispatch of generation to meet demand at scale.

Most commonly in the U.S., residential and commercial PV systems sized at or beneath $10$ kW - $30$ kW have a simplified application process~\cite{bird2018review} that may not capture the PV characteristics required to model such DERs. As more of these systems are connected to the grid, modeling the power they generate is of growing importance for issues ranging from energy markets and dispatch to grid power flow control. The estimation and forecast of PV power are made difficult as only a minority of systems continuously report their generation in a publicly accessible manner. Knowledge of PV system characteristics is required in the different regional PV modeling approaches used to reconstruct the missing power measurements~\cite{killinger2018search}.

Overall, regardless of scale, a careful assessment of PV generation potential, trends, and performance in a local area is a prerequisite to designing, tuning, and operating PV-based utility systems confidently and effectively. Recreating PV production data can also become a key activity for the areas where PV data available to the utility is scarce. In an effort to address these needs, utilities can identify BTM and FTM PVs already installed in a given area and measure PV generation. Utilities may also request that PV owners allow them to place PV meters on systems identified at strategic grid interconnection locations. The metadata inferred from the measurements obtained from such meters may reveal the generation patterns of the owners, as well as the neighboring sites with similar PV system properties. These measurement systems are an extension of the utility's field telemetry and can become indistinguishable from SCADA.

There are several ways to collect PV output and metadata depending on the type of PV system, the specific needs of the user, and if the utility can acquire special permission to collect the PV data directly. As previously mentioned, an obvious source of PV production data is the smart meter. Smart meters measure the net loads from which PV generation profiles can be extracted. One caveat of using smart meter data for any inferential analysis of the PV metadata is that the data needs to be disaggregated first.  The other monitoring sources that may readily provide sufficient time series data for such purposes are PV meters, data loggers, and smart inverters. PV meters ~\cite{montano-martinez2021detailed} and data loggers ~\cite{kipp2018solar} can be installed in the electrical panel of a PV system to gather instantaneous generation data produced by solar panels. Data loggers can also measure the temperature and humidity that characterize the outputs. This can be used to regularly track the performance of the PV system over time~\cite{zipp2018when}. smart inverters with built-in data logging capabilities have been proven to be successful in closely monitoring PV performance and the operating and technical factors that influence PV production.

This paper aims to study a DP randomized mechanism that would allow a utility to publish solar PV metadata to untrusted third-party analysts while protecting customer privacy. Using a solar power generation dataset, the paper presents an algorithm that utilizes solar generation in congruence with public local solar irradiation patterns to perform DP Bayesian optimization to infer metadata such as surface azimuth, surface tilt, and nameplate capacity.

The rest of the paper is organized as follows.  \Cref{sec:Preliminaries} describes the solar metadata inference problem and proposes a grid-search-based mechanism to infer the metadata. In  \Cref{sec:Methodology}, we propose Bayesian optimization (BO) as a tool to scale the inference of the metadata and introduce DP and propose a DP BO-based mechanism to publish the inferred solar PV metadata.  \Cref{sec:Numericals} discusses the numerical results and performance of the proposed mechanism. Finally, we conclude in \Cref{sec:Conclusion}.

\section{Solar PV Metadata Inference}\label{sec:Preliminaries}

\begin{figure}
	\centering
	\includegraphics[width=0.45\textwidth]{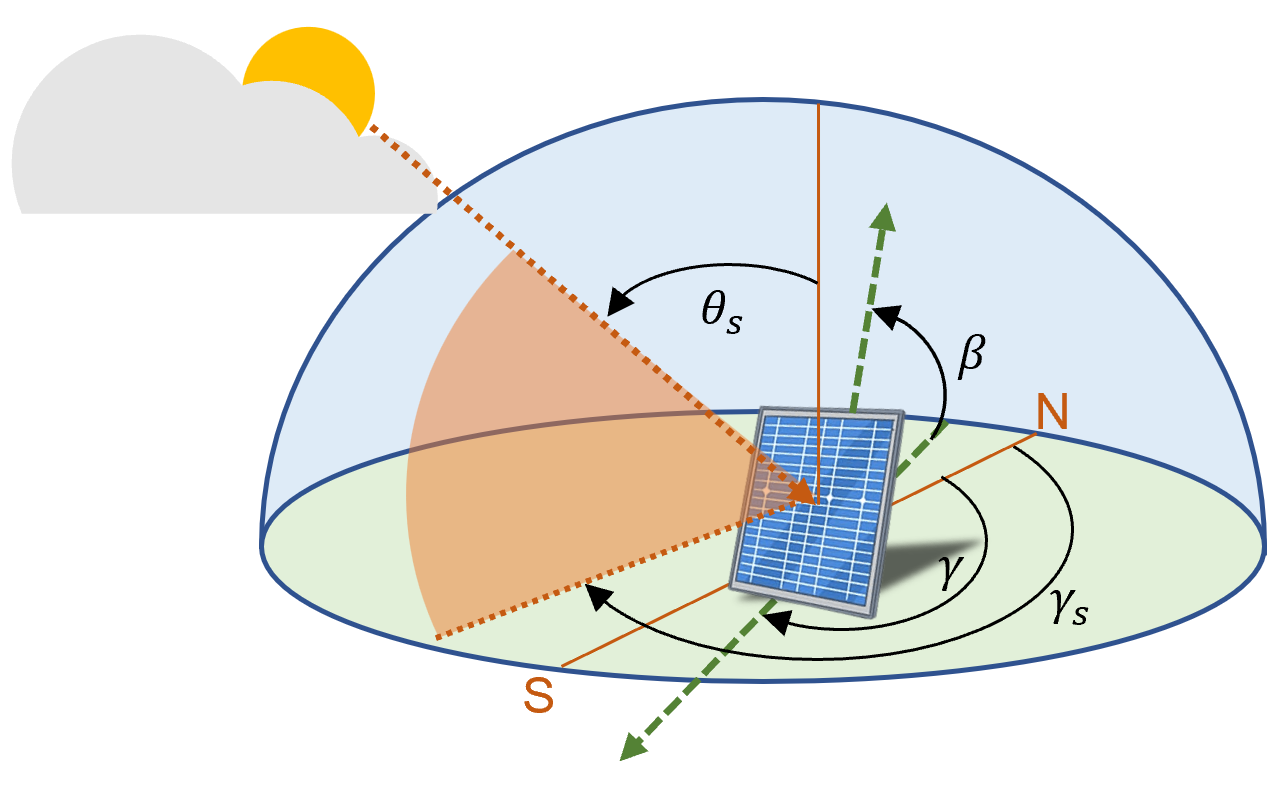}
	\caption[Metadata Illustration]{Illustration of the various angles involved in modeling solar power generation.}
	\label{fig:metadata}
\end{figure}
The availability of accurate metadata about the various distributed solar PV systems in the field is essential to perform grid operations and analyses. This metadata includes the following aspects of a PV system:
\begin{itemize}[leftmargin=*]
	\item \textit{\underline{\textbf{Nameplate capacity:}}} The amount of power the panel can output while it is running in ideal conditions over some duration.
	\item \textit{\underline{\textbf{Surface azimuth:}}} The angle the solar panels are facing; it is measured in a clockwise direction from the north -- i.e., a panel facing north (respectively east, south, and west) is said to have a surface azimuth of $0\degree$ (respectively $90\degree$, $180\degree$, and $270\degree$).
	\item \underline{\textbf{\textit{Surface tilt:}}} The angle the solar panel makes with the horizontal; it should be set to a value between $0\degree$ and $90\degree$.
\end{itemize}
These parameters are illustrated in \cref{fig:metadata} where the surface azimuth and tilt are denoted by $\gamma$ and $\beta$, respectively. 

\subsection{Solar Power generation as a function of metadata}\label{sec:solar_generation}
In this section, we will briefly describe a prevalent solar generation model by NREL~\cite{dobos2014pvwatts} and present preliminary data analysis that indicates the sensitive nature of a PV system's nameplate capacity, surface azimuth and surface tilt.

For a fixed-tilt solar panel, the angle of incidence, $\theta$, of solar rays is given by:
\begin{equation}
	\theta=\cos^{-1}\left(\sin\theta_s \cos(\gamma-\gamma_s)\sin\beta+\cos\theta_s cos\beta\right) \label{eq:ang_incidence}
\end{equation}
where $\theta_s, \gamma, \gamma_s$ and $\beta$ are the solar zenith, surface azimuth, solar azimuth, and surface tilt angles, respectively~\cite{dobos2014pvwatts}. The transmitted irradiance on the plane-of-array $I_t[k]$ at time $k = 0,\ldots , n-1$, is:
\begin{equation}
	I_t[k]=f(\theta)\cos\theta I_n[k]+I_{ds}[k]+I_{dg}[k]\label{eq:poa}
\end{equation}
where $I_n[k], I_{ds}[k], I_{dg}[k]$ are the normal, sky diffuse, and ground diffuse irradiance, respectively, and $f(\theta)$ is an attenuation factor that depends on the characteristics of the glass in the panel.
Refer to \cref{fig:metadata} for an illustration of the various angles discussed here.

Finally, the AC power generated by the panel is given by:
\begin{equation}
	p[k] = \zeta\overline{p}(1+\kappa(\tau_c[k]-\tau_r)) I_t[k],\label{eq:acpower}
\end{equation}
where $\overline{p}$ is the \textit{nameplate capacity} of the panel, $\kappa$ is a temperature coefficient, $\tau_c$ and $\tau_r$ are the cell and reference temperatures of the panel, respectively, and $\zeta$ represents the other confounders not explicitly modeled. From \crefrange{eq:ang_incidence}{eq:acpower}, we establish the relationship between the metadata (consisting of nameplate capacity and surface azimuth and tilt) and the solar power generated. \Cref{eq:acpower} also establishes that the power generated is directly proportional to the global transmitted on the plane-of-array.
In addition to the dependence on the surface azimuth and tilt, various other seasonal influencers such as temperature, relative air mass, wind speed and direction, rainfall, the strength of the irradiance, etc., also affect solar power generation. A few of these confounders are modeled explicitly in \cref{eq:acpower}, while the rest are lumped into the multiplicative $\zeta$ factor.

This model naturally leads to the following conjectures:
\begin{conjecture}
	The plane-of-array irradiance is strongly correlated to the solar power generated by the PV system.
\end{conjecture}
\begin{conjecture}\label{conj:metadata}
	Any changes in surface azimuth and tilt result in a significant change in solar power generation.
\end{conjecture}
\begin{conjecture}
	Various seasonal confounders affect solar power generation.
\end{conjecture}

\subsection{Solar PV Metadata Is Sensitive Information}

\begin{figure*}[!htbp]
	\centering
	\subfloat[][]{
		\includegraphics[width=0.32\textwidth]{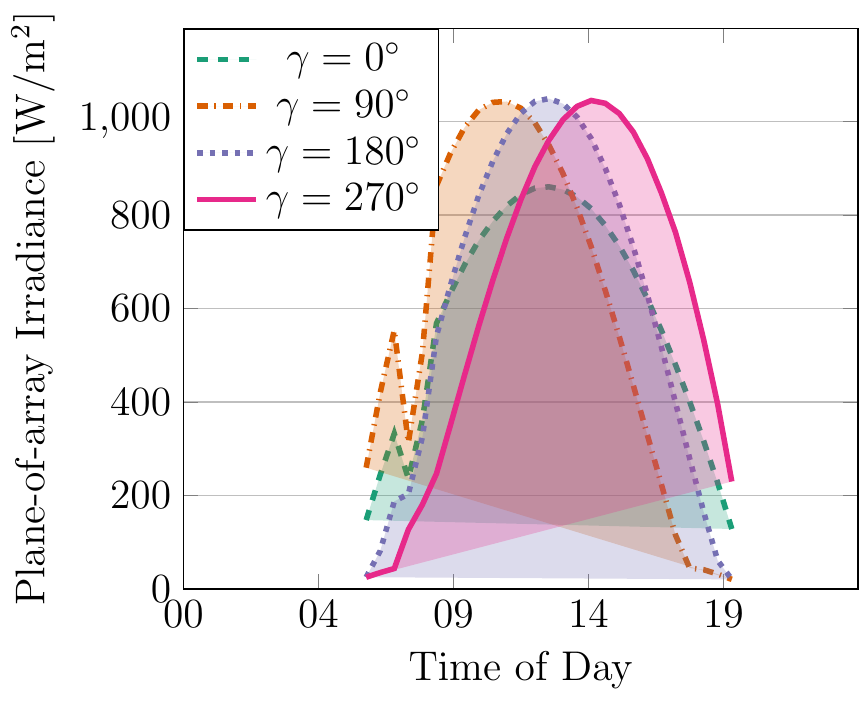}\label{fig:var_az_same_tilt}}\quad
	\subfloat[][]{
		\includegraphics[width=0.32\textwidth]{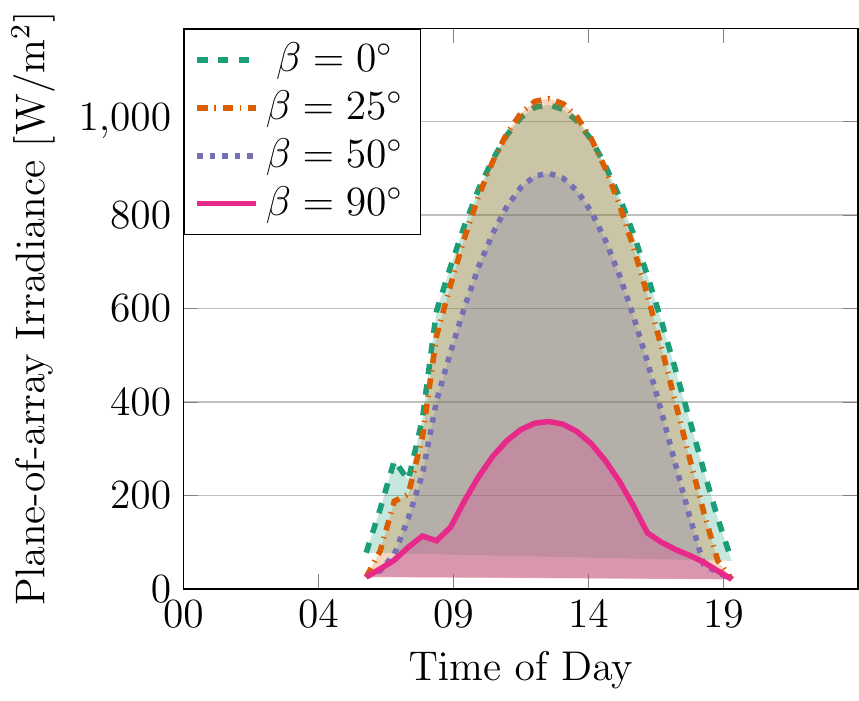}\label{fig:same_az_var_tilt}}\quad
	\subfloat[][]{
		\includegraphics[width=0.3\textwidth]{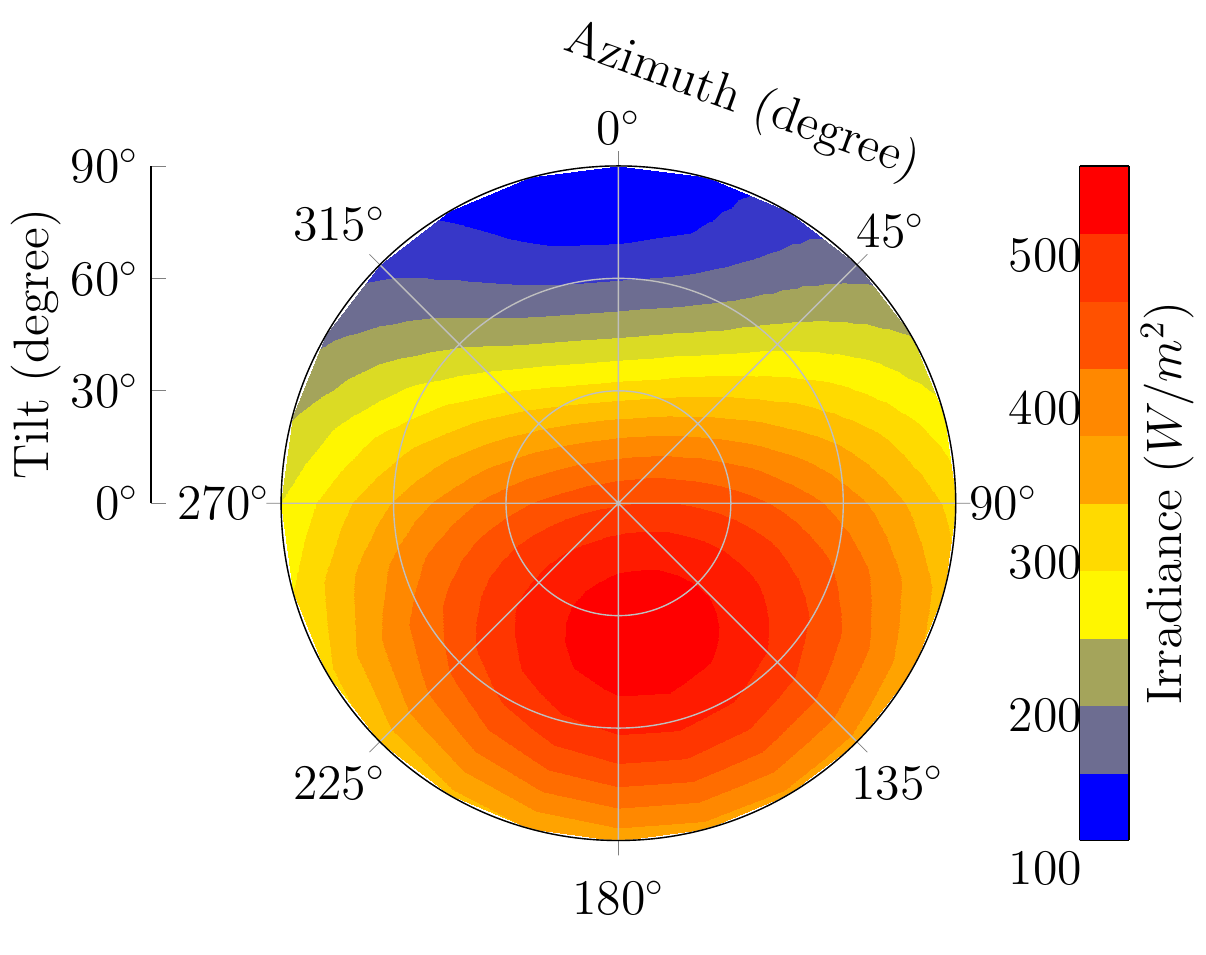}\label{fig:yearly_avg}}
	\caption{For a location in California (34.0122,-117.6889): (\subref{fig:var_az_same_tilt}) The plane-of-array irradiance on July 15, 2016, for various azimuths ($\gamma$) at a fixed tilt of $\beta=25\degree$. (\subref{fig:same_az_var_tilt}) The plane-of-array irradiance for various tilts ($\beta$) at a fixed azimuth of $\gamma=180\degree$ on the same day. (\subref{fig:yearly_avg}) The daily mean plane-of-array irradiance for the year $2016$ averaged over the course of the year. In this polar plot, the radial and the angular axes represent the surface tilt and azimuth, respectively.}
	\label{fig:az_tilt}
\end{figure*}
The aforementioned conjectures are further demonstrated in \cref{fig:az_tilt}, where the global irradiance over the course of a summer day in Chino, California is plotted for various azimuths at a fixed tilt of $25\degree$ (see \cref{fig:var_az_same_tilt}) and for various tilts at a fixed azimuth of $180\degree$ (see \cref{fig:same_az_var_tilt}). As seen in these figures, as the azimuth and tilt are varied, the solar generation profile of a system drastically changes in both the time of peak generation and the magnitude of the peak generation.  

Furthermore, \cref{fig:yearly_avg} shows the mean daily average irradiation on a polar plot with the surface azimuth and tilt on the angular and radial axes, respectively. As is apparent from the plot, the effectiveness of a solar PV system in generating energy directly depends on its surface azimuth and tilt. Particularly, for a system in Chino, California, this plot indicates that a panel with azimuth between $165\degree$ and $180\degree$, and a tilt between $35\degree$ and $40\degree$ maximizes the expected average daily power. 

Finally, we compiled the peak daily solar generation from several households in California, and the kernel density estimates of nameplate capacity under four quantiles of the maximum power generated are shown in \cref{fig:peak_daily}. This figure leads to the conclusion that generation profiles are sensitive to nameplate capacity.

\begin{figure}
	\centering
	\includegraphics[width=0.45\textwidth]{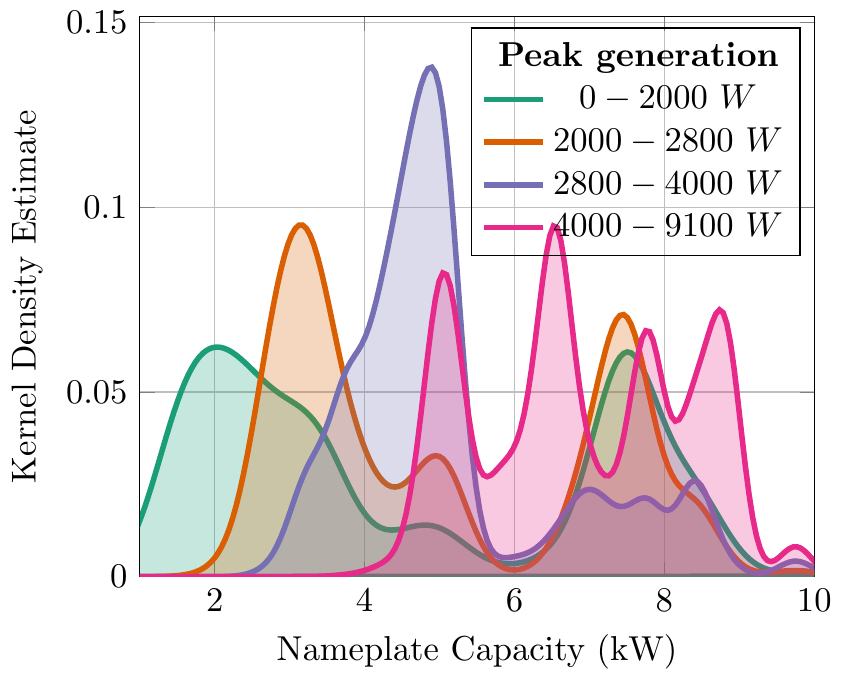}
	\caption{Joint plot of nameplate capacity (x-axis) and daily peak power generation in Watts (y-axis) with the systems classified based on four quantiles of daily peak power generated.}
	\label{fig:peak_daily}
\end{figure}
% \begin{figure*}
% 	\centering
% 	\includegraphics[width=0.95\textwidth]{./figs/monthly_top.png}
% 	\caption{The top $1\%$ $(\beta, \gamma)$ pairs with the lowest $\ell$ for each month showcasing the seasonal effects of the various confounders that affect solar power generation.}
% 	\label{fig:monthly_results}
% \end{figure*}

These figures lend credence to \Cref{conj:metadata} that nameplate capacity and surface azimuth and tilt are highly indicative of the generation profiles of solar PV systems, making each of these PII. 

\subsection{Solar PV Metadata Inference}\label{sec:solar_pv_metadata_inference_sketch}
The previous section established that the solar generation profiles are highly correlated with the local global solar irradiation profiles. By denoting the surface azimuth and tilt pair using $\blambda = [\gamma, \beta]^\top$ and their combined domain using $\bLambda$, this correlation can be used to find the candidate $\blambda \in \bLambda$ whose solar generation curve matches the solar irradiation profile the best. In other words, we want to solve the following maximization problem:
\begin{align}\label{eq:objective:solar_pv_metadata_inference}
	\max_{\blambda \in \bLambda} \quad \ell_{\bm{p}}^{(\mathrm{loc})}(\blambda) \qquad \equiv \qquad \max_{\blambda \in \bLambda} \quad - \|\hat{\bm{p}}_{\blambda}^{(\mathrm{loc})} - \bm{p}\|^2_2,
\end{align}
where $\bm{p} = [p[0], \ldots, p[K-1]]^\top$ is the vector of AC power generated by a Solar PV system from time $k=0$ to $k=K-1$, and $\hat{\bm{p}}_{\blambda}^{(\mathrm{loc})} = [\hat{p}_{\blambda}^{(\mathrm{loc})}[0], \ldots, \hat{p}_{\blambda}^{(\mathrm{loc})}[K-1]]^\top$ with $\hat{p}_{\blambda}^{(\mathrm{loc})}[k]$ the potential total plane of array irradiation generated by a system with the metadata $\blambda$ at time $k$ at the same location, denoted by ${(\mathrm{loc})}$. Throughout the rest of the paper, the location indicator will be omitted for the sake of brevity.

Note that $\hat{p}_{\blambda}[k]$ is calculated according to the model presented in \cref{eq:poa} using the publicly available solar irradiation data available at National Renewable Energy Laboratory's National Solar Radiation Database (NSRDB)~\cite{sengupta2018national}.

There are several challenges to implementing this simple problem:
\begin{itemize}[leftmargin=*]
	\item Various highly stochastic and difficult-to-model factors (such as diffusion, cloud cover, shading, dust, etc.) attenuate the total solar irradiation before it reaches the solar panel. One possible solution to overcome this issue is to find the $\blambda\in \bLambda$ whose solar irradiation profile matches the clearest day's solar generation profile, where $\bLambda$ is the set of all possible values that $\blambda$ can take.
	\item However, as solar power generation is seasonal, one cannot solely depend on any single day's solar generation profile to make inferences. We thus group days (e.g., into months or weeks), obtain the clearest days in each group, and find the $\blambda$ that best fits the generation profiles of all or most of the chosen days.
	\item Finally, the total number of possible values that $\blambda$ can take is very high. For example, to search in the integer space for azimuth and tilt, one would have to consider $360~\times~90~=~32,400$ possibilities. This would mean querying NSRDB for the solar irradiation data $32,400$ times for each group, which is not scalable.
\end{itemize}

With these caveats in mind, a methodology to infer $\blambda$ from generation data $p[k]$ for all $k\in[N-1]$ would take the following steps:
\subsubsection{Normalization}
In order to remove the dependence of $\hat{\bm{p}}_{\blambda}$ and $\bm{p}$ on the nameplate capacity, we first normalize the daily generation data and the potential daily generation data to be between 0 and 1.
\subsubsection{Prototypical Generation profile}
As discussed earlier in this section, owing to the seasonal nature of solar generation, we can not solely depend on any single day's solar generation profile. Thus, we first split the solar generation time series of a PV system into groups of contiguous time series corresponding to months or quarters. Let the set of these groups be denoted by $\mathcal{G}$. We select the day in each group, $g\in \mathcal{G}$, which has the highest correlation between its direct normal irradiance (DNI) and its total global horizontal irradiance (GHI). These days represent the clearest days in their corresponding groups, and their power generation is denoted by $\bm{p}_g$.  
\subsubsection{Grid Search} 
For each group, find the distance between its prototypical generation, $\bm{p}_g$, and the potential plane-of-array irradiation, $\hat{\bm{p}}_{\blambda}$, for all $\blambda \in \bLambda$, i.e.:
\begin{equation}
	\loss_g(\blambda) := -\frac{1}{N}\|\hat{\bm{p}}_{\blambda} - \bm{p}_g\|^2_2,
\end{equation}
where $N$ is the number of measurements available for a day (for example, it is $24$ for data with a one-hour resolution).
\subsubsection{Fit Score}
Finally, for each $\blambda \in \bLambda$, we calculate a fit score as follows:
\begin{equation}\label{eq:fit_score}
	\loss(\blambda) = \sum_{g\in\mathcal{G}} \loss_g(\blambda).
\end{equation}
The best fit $\sblambda$ is then chosen as the $\blambda \in \bLambda$ with the highest fit score. 
However, as these steps do not scale well as the size of $\bLambda$ grows, we turn to Bayesian optimization (BO), described next, to infer $\sblambda$.

\section{Proposed Methodology: DP-BO Solar Metadata Inference}\label{sec:Methodology}
This section presents a methodology to infer the surface metadata of distributed PV systems from their solar generation data in a differentially private manner to protect the DP of the system's solar generation profiles. The inference is mainly facilitated via Bayesian optimization to maximize the correlation between:
\begin{enumerate}
	\item the solar generation profiles, and
	\item the solar irradiation patterns for candidate metadata values corresponding to the location and time of year.
\end{enumerate}
We consider the scenario where the solar generation data pertaining to consumers is either housed at a central trusted curator such as a utility or at the consumers themselves. An untrusted third-party analyst may then query the trusted curator or the consumers for statistics about the datasets, to which the curator responds with a DP answer as illustrated in \Cref{fig:trusted_curator}. 

Given a solar generation profile, $\bm{p}$, our goal is to maximize an unknown function $\loss: \bLambda \rightarrow \mathbb{R}$:
\begin{equation}
	\max_{\blambda \in \bLambda} \quad \loss(\blambda), \label{eq:bayes_opt}
\end{equation}
where $\loss$, given by \cref{eq:fit_score}, is the fit score of the solar generation, $\bm{p}$, when compared to the irradiation patterns with the parameters $\blambda \in \bLambda \subseteq \mathbb{R}^2$. This problem has the following characteristics:
\begin{enumerate}
	\item The function $\loss(\cdot)$ is expensive to evaluate, see \Cref{sec:solar_pv_metadata_inference_sketch}.
	\item The size of the input $\blambda$ is not too large. In our problem, it is $2$.
	\item The feasible set $\bLambda$ is a simple, finite set whose membership is easy to establish. In our case, the feasible set can be characterized by the following equation:
	      \begin{equation}
		      \bLambda = \{\blambda\equiv(\gamma, \beta) \mid \gamma \in [0, 360), \beta\in[0, 90]\}.
	      \end{equation}
\end{enumerate}
These characteristics make Bayesian optimization a good tool to solve the aforementioned problem.

\subsection{Bayesian Optimization}\label{sec:BayesOpt}

Bayesian Optimization (BO) is an iterative algorithm used to evaluate the optimizer of difficult/expensive-to-evaluate functions. Formally, BO considers the problem of maximizing the function $\loss(\cdot)$ over its domain $\bLambda$, such that a sequence $\bLambda_T := [\blambda_1, \ldots, \blambda_{T}]$ is chosen to maximize the sum $\sum_{t=1}^T \loss(\blambda_t)$. 
By defining the instantaneous regret of choosing $\blambda_t$ at time $t$ as $r_t := \loss(\sblambda) - \loss(\blambda_t)$ and the cumulative regret as $R_T := \sum_{t=1}^T r_t$, where $\sblambda$ is the maximizer, the above problem is equivalent to minimizing $R_T$.

\begin{algorithm}
	\caption{Bayesian optimization}\label{alg:BayesOpt}
	\KwData{Prior over $\loss$}
	\KwResult{$\sblambda \in \bLambda$}
	\textbf{Warm Start:} Observe the objective function, $\loss$, at a few points, say at $[\blambda_0, \ldots, \blambda_{t-1}]$.

	\For{$\tau = t\ldots T$}{
		\textbf{Update Posterior}: Update the posterior with the data observed so far.

		\textbf{Find the next point}: Utilizing the updated posterior distribution and the acquisition function, evaluate the point, $\blambda_{\tau}$, to sample next.

		\textbf{Observe}: Evaluate $\loss$ at $\blambda_{\tau}$.
	}

	\textbf{Optimizer}: Choose the point whose evaluation was the highest.
\end{algorithm}

A sketch of the algorithm is showcased in \Cref{alg:BayesOpt}. BO approaches this problem by modeling $\loss$ to be a random process and places a prior belief on it. In each iteration, the algorithm observes the function at a selected point $\blambda_{t}$, and updates the posterior distribution of the function. Using the updated posterior distribution, the algorithm then uses a special acquisition function to evaluate the next point at which to observe the function. This is repeated until the observation budget is exhausted. In the rest of the section, the aforementioned steps are described in detail.

\subsubsection{Gaussian Process Prior – Prior Distribution over the objective function}\label{sec:BayesOpt_GP}
As discussed above, BO places a prior on the objective function, whose posterior is updated after each subsequent observation of a carefully chosen point. In this section, we describe the Gaussian Process (GP) prior.

Suppose that before iteration $t$, we have observed the function evaluated at points $\bLambda_t := [\blambda_1, \ldots, \blambda_{t}]$ and the function evaluations at these points are ${\bloss}_t = [{\loss}_1, \ldots, {\loss}_t]$. We place a multivariate normal prior over ${\bloss}_t$ with mean $\bm{\mu}_t$, that captures our estimate of the function across the domain, and covariance $\bm{\Sigma}_t$, that captures our uncertainty of the estimate. That is:
\begin{equation}
	{\bloss}_t(\bLambda_t) \sim \mathcal{N}(\bm{\mu}_t, \bm{\Sigma}_t),
\end{equation}
where the mean vector, $\bm{\mu}_t$, is the vector of mean function ($\mu$) evaluated at $\bLambda_t$, i.e., 
\begin{equation}
	\bm{\mu}_t := \bm{\mu}(\bLambda_t) := [\mu(\blambda_1), \ldots, \mu(\blambda_t)].
\end{equation}
Similarly, the covariance matrix is such that 
\begin{equation}
	[\bm{\Sigma}_t]_{\tau, \tau'} := [\bm{\Sigma}(\bLambda_t, \bLambda_t)]_{\tau, \tau'} := \Sigma(\blambda_{\tau}, \blambda_{\tau'}),	
\end{equation}
where $\Sigma$ is a kernel function. 

Without loss of generality, through the rest of the paper, we assume that $\mu(\blambda)=0,~\forall \blambda \in \bLambda$~\cite{rasmussen2005gaussianb}, and a Radial Basis Function kernel (RBF), i.e.,
\begin{equation}
	\Sigma(\blambda_{t}, \blambda_{t'}) = \exp\left(-\|\blambda_{t} - \blambda_{t'}\|^2\right) \leq 1.
\end{equation}
Note that $\Sigma(\blambda_{t}, \blambda_{t})=1$. Furthermore, after observing ${\bloss}_t$, the posterior mean and variance of the objective function evaluated at a point $\blambda$ is easily computed using the following equations:
\begin{subequations}\label{eq:posterior}
	\begin{align}
		\loss(\blambda)    & \mid {\bloss}_t \sim \mathcal{N}(\mu_t(\blambda), \sigma_t^2(\blambda))                                                                           \\
		\mu_t(\blambda)      & = \bm{\Sigma}(\blambda, \bLambda_t)^\top\bm{\Sigma}(\bLambda_t, \bLambda_t)^{-1} {\bloss}_t \label{eq:BO_mu_t}                                                          \\
		\sigma_t^2(\blambda) & = 1 - \bm{\Sigma}(\blambda, \bLambda_t)^\top\bm{\Sigma}(\bLambda_t, \bLambda_t)^{-1}\bm{\Sigma}(\blambda, \bLambda_t),
	\end{align}
\end{subequations}
where $\bm{\Sigma}(\blambda, \bLambda_t) := [\Sigma(\blambda, \blambda_1), \ldots, \Sigma(\blambda, \blambda_t)]^\top$.

\subsubsection{Acquisition Function – Finding the next point}
As mentioned above, after updating the posterior, BO uses an acquisition function to select the next point that the algorithm observes. In this paper, we select $\blambda$ maximizes the upper-confidence bound (UCB) of the posterior GP model of $\loss$:
\begin{equation}
	\blambda_t := \arg\max_{\blambda \in \bLambda} \mu_{t-1}(\blambda) + \sqrt{\phi_t}\sigma_{t-1}(\blambda),\label{eq:ucb}
\end{equation}
where $\phi_t$ is a parameter that trades off the exploitation of maximizing $\mu_{t-1}(\blambda)$ and the exploration of maximizing $\sigma_{t-1}(\blambda)$. The exact manner in which $\phi_t$ is chosen is described later in \cref{sec:Methodology}, but it is an increasing function in $t$ and, as evident in the \cref{eq:ucb}, captures the number of standard deviations in the upper confidence bounds.

\subsubsection{Regret Bounds for GP-UCB}
If the domain of the optimization, $\bLambda$ is of finite size $|\bLambda|$, the GP-UCB algorithm described in \cref{alg:BayesOpt} achieves sublinear regret bounds. The proof of the said sublinear bound was facilitated by first proving the following lemma~\cite[Lemma 5.1]{srinivas2012informationtheoretic}:
\begin{lemma}\label{lem:gp-ucb-bound}
	Running GP-UCB with $\phi_t$ for a function $\loss$ sampled from a GP with mean function zero and covariance function $\Sigma(\blambda, \blambda')$, we obtain the following bound for all $t\geq 1$:
	\begin{equation}
		Pr\{ |\loss(\blambda) - \mu_{t-1}(\blambda)| \leq \sqrt{\phi_t} \sigma_{t-1}(\blambda)\} \geq 1-\delta,
	\end{equation}
	where $\phi_t = 2\log(|\bLambda|t^2\pi^2/6\delta)$ and $0<\delta\lll 1$.
\end{lemma}
This lemma is sufficient for our purposes. Thus far, we have described a methodology to infer the metadata. Next, we describe the DP and the DP-BO metadata inference mechanism.

\subsection{Differential Privacy}\label{sec:dp_prelims}
In this section, we briefly introduce the conventional definitions that explain how DP is measured and established. As discussed in \cite{ravi2022differentially}, the `15/15' Rule~\cite{cpuc2011proposed}, an often used `anonymization' technique in the industry, does not actually provide any statistical guarantees of privacy. We propose the use of DP as a technically sound alternative to this rule in the smart grid industry to publish private aggregate query answers.

To define DP, consider an analyst with a query $\bblambda(\cdot)$ on a dataset $\bm{p}$. 
Informally, to say that an algorithm (often referred to as the DP mechanism) protects the differential privacy of the individual records in the dataset $\bm{p}$, it should have a random output, denoted by $\tbblambda(\bm{p})$ with distribution $f(\tblambda|\bm{p})$\footnote{This is the probability density function for continuous random queries and the probability mass function for discrete random variables}, should approximately be the same as $\bblambda(\bm{p})$, with or without the data of any specific individual record. In other words, the randomness of $\tbblambda$ should be enough so that an observed output from it will not reveal whether one of two datasets $\bm{p}$ or $\bm{p}'$, differing only in one record, was the input to $\tbblambda$. The rationale for this is that if the untrusted analyst is unable to tell which input the output came from, then they can not infer the presence or lack thereof of any one particular record and the record's content.

In formal terms, differential privacy was first introduced in \cite{dwork2006calibrating, dwork2006our}. It states that:
\begin{definition}[$(\epsilon,\delta)$-Differential privacy]\label{def:DP}
	A randomized mechanism $\tbblambda(\bm{p})$ is $(\epsilon,\delta)$-differentially private if for all neighboring datasets $\bm{p}$ and $\bm{p}'$ that differ in one point, for any arbitrary event pertaining to the outcome of the query, the randomized mechanism satisfies the following inequality
	\begin{equation}\label{eq:e-dp}
		\forall \bLambda,\quad
		Pr(\tbblambda(\bm{p}) = \blambda) \leq \exp(\epsilon)Pr(\tbblambda(\bm{p}') = \blambda) + \delta,
	\end{equation}
	where $Pr(\mathcal{A})$ denotes the probability of the event $\mathcal{A}$, for some privacy budget $\epsilon\geq 0$ and $\delta \in [0,1]$.
\end{definition}
In this definition, $\epsilon$ is known as the privacy budget. It, along with $\delta$, dictates the level of DP that is guaranteed by a mechanism. 

% Furthermore, in dealing with a multidimensional answer such as the one we are interested in, a first common simplification is to use independent mechanisms for different components; a common simplification is to use the following lemma to map the $(\epsilon_j,\delta_j)$ results for the scalar independent mechanism employed onto a global $(\epsilon,\delta)$ result:
% \begin{lemma}[Sequential composition~\cite{dwork2014algorithmic}]\label[lemma]{lem:series_composition}
% 	If $n$ randomized algorithms $M_i$, $i=1,2,\ldots,n$, are $(\epsilon_i,\delta_i)$-DP, then the sequential execution of these algorithms on the database $\bm{p}$ provides $(\sum_i \epsilon_i, \sum_i \delta_i$)-differential privacy.
% \end{lemma}

Some of the most ubiquitous DP mechanisms are additive noise mechanisms such as the Laplace and the Gaussian Mechanisms~\cite{dwork2006calibrating}. In these mechanisms, additive random noise is added to the true query answer. A key challenge here is to calibrate the amount of noise that is needed to satisfy the definition of DP, but not overdo it such that the DP answer becomes useless. However, the task here is to not answer a numerical query, such as the count or the average. We are trying to choose the best $\blambda \in \bLambda$ that maximizes the fit score in \cref{eq:fit_score}. An exponential mechanism, first defined by McSherry and Talwar~\cite{mcsherry2007mechanism}, is ideal for this purpose.
\begin{definition}[Exponential Mechanism]\label{def:exponential_mechanism}
	Given a set $\bLambda$ of possible outputs, a score function $\loss: \bLambda \rightarrow \mathbb{R}$ with a global sensitivity of $\Delta_{\ell}$, the exponential mechanism outputs $\blambda \in \bLambda$ with probability proportional to: 
	\begin{equation}\label{eq:exponential_mechanism}
		\exp\left(\frac{\epsilon\loss(\blambda)}{2\Delta_{\ell}}\right) \times \pi(\blambda),
	\end{equation}
	where $\pi$ is a prior measure on $\bLambda$.
\end{definition}
Furthermore, McSherry and Talwar proved that this exponential mechanism $\epsilon$-DP.
\begin{theorem}[Exponential Mechanism is $\epsilon$-DP~\cite{mcsherry2007mechanism}]\label{thm:exp_mechanism_is_eDP}
	The exponential mechanism in \Cref{def:exponential_mechanism} is $\epsilon$-DP.
\end{theorem}
DP is not a property of the data, but of the mechanism, and more specifically the query that the mechanism is privatizing. As such, the amount of noise is often chosen proportional to the sensitivity of the query defined below:
\begin{definition}[Global Sensitivity]
	The global sensitivity of a function $\loss: \bLambda \rightarrow \mathbb{R}$ over all neighboring datasets $X, \bm{p}'$ is:
	\begin{equation}
		\Delta_{\ell} := \max_{\bm{p},\bm{p}' \subseteq \mathcal{X}, \blambda \in \bLambda} |\loss(\blambda) - \ploss(\blambda)|
	\end{equation}
\end{definition}
This definition establishes that for any two neighboring datasets, $\bm{p}$ and $\bm{p}'$, $\bblambda(\bm{p})$ and $\bblambda(\bm{p}')$ differ by at most $\Delta_{\bblambda}$. 

With these definitions, we are now ready to infer Solar PV metadata in a differentially private manner with the BO objective as our score function and an exponential mechanism.

\subsection{DP Metadata Inference based on BO} 
In this section, we use the various sub-processes, namely -- 1) Solar metadata inference mechanism described in \Cref{sec:solar_pv_metadata_inference_sketch}, 2) BO for solar PV metadata inference mechanism described in \Cref{sec:BayesOpt}, and 3) DP described in \Cref{sec:dp_prelims} -- to describe a mechanism to differentially privately publish the Solar PV metadata. 

The algorithm includes the following main subroutines:
\subsubsection{Preprocess} In this step, the solar generation profile is preprocessed in order to prepare it for the BO step. The preprocessing steps, described in \Cref{sec:solar_pv_metadata_inference_sketch}, include normalization, identification of monthly prototypical solar generation days and their generation profiles, and a facility to obtain solar irradiation patterns for the said location on these days.
\subsubsection{BO} Using the preprocessed data, we obtain an estimate of $\loss$ using the algorithm described in \Cref{sec:BayesOpt} and listed in \Cref{alg:BayesOpt}.
\subsubsection{DP} Finally, we use the exponential mechanism to publish the metadata in a differentially private manner.

This set of routines is listed in \Cref{alg:DP}:
\begin{algorithm}
	\caption{DP-BO Solar PV Metadata inference}\label{alg:DP}
	\KwData{$\bm{p}$; $T$; $\bLambda\subseteq \mathbb{R}^2$, $(\epsilon, \delta)$; $\sigma^2_0$; $\rho_T$; $\mu_0 = 0$}
	\KwResult{$\tblambda$}
	\For{t = 1\ldots T}{
	$\phi_t \gets 2\log(|\bLambda|t^2\pi^2/(6\delta))$

	$\blambda_t \gets \arg \max_{\blambda \in \bLambda} \mu_{t-1}(\blambda) + \sqrt{\phi_t} \sigma_{t-1}(\blambda)$

	Observe fit score, $\loss_t$, given $\blambda_t$

	Update $\mu_{t}$ and $\sigma^2_{t}$ according to \cref{eq:posterior}
	}
	Draw $\tilde{\blambda} \in \bLambda$ w.p. $Pr[\blambda] \propto \exp\left( \frac{\epsilon\mu_T(\blambda)}{2(2\sqrt{\phi_{T+1}}+c)}\right)$\;
\end{algorithm}

\subsection{DP guarantees of DP-BO}
In this section, we provide the DP guarantees of DP-BO, listed in \cref{alg:DP}, and a sketch of its proof.
The main theorem states that DP-BO is $(\epsilon, \delta)$-DP:
\begin{theorem}[\Cref{alg:DP} is $(\epsilon, \delta)$-DP]
	Let $\bblambda_{\ell}(\bm{p})$ denote the mechanism presented in \Cref{alg:DP}. If \cref{lem:gp-ucb-bound} holds, then the output of $\bblambda_{\ell}(\bm{p})$, $\tblambda$, is $(\epsilon, \delta)$-DP for any pair of neighboring datasets $\bm{p}$ and $\bm{p}' \in \mathcal{X}$.
\end{theorem}
For brevity, the proof of the above theorem is relegated to \Cref{app:proof} and we provide a brief sketch of the proof here.
\subsubsection*{Sketch of the proof}
In order to prove that the mechanism described in \Cref{alg:DP} is differentially private, we first establish a bound for the global sensitivity of the fit score function $\loss$ in \Cref{lem:sensitivity_bound} (see \Cref{app:proof}). Facilitated via \Cref{lem:gp-ucb-bound}, with a high probability of $1-\delta$ for a very small $\delta$, \Cref{lem:sensitivity_bound} establishes a bound for the global sensitivity of the BO estimate of the fit score function, $\mu_t(\cdot)$ at iteration $t>1$. Following this, using the definition of DP in \Cref{def:DP}, in \Cref{lem:exponential_mechanism_proof}, we prove that the publication of $\tblambda$ via the exponential mechanism in \cref{def:exponential_mechanism} with $\mu_t$ as the score function is $(\epsilon, \delta)$-DP for any pair of neighboring datasets $\bm{p}$ and $\bm{p}' \in \mathcal{X}$. 

\section{Numerical Results and Observations for Solar Metadata Inference}\label{sec:Numericals}
\begin{figure*}[!htbp]
    \centering
    \subfloat[][]{
        \includegraphics[width=0.3\textwidth]{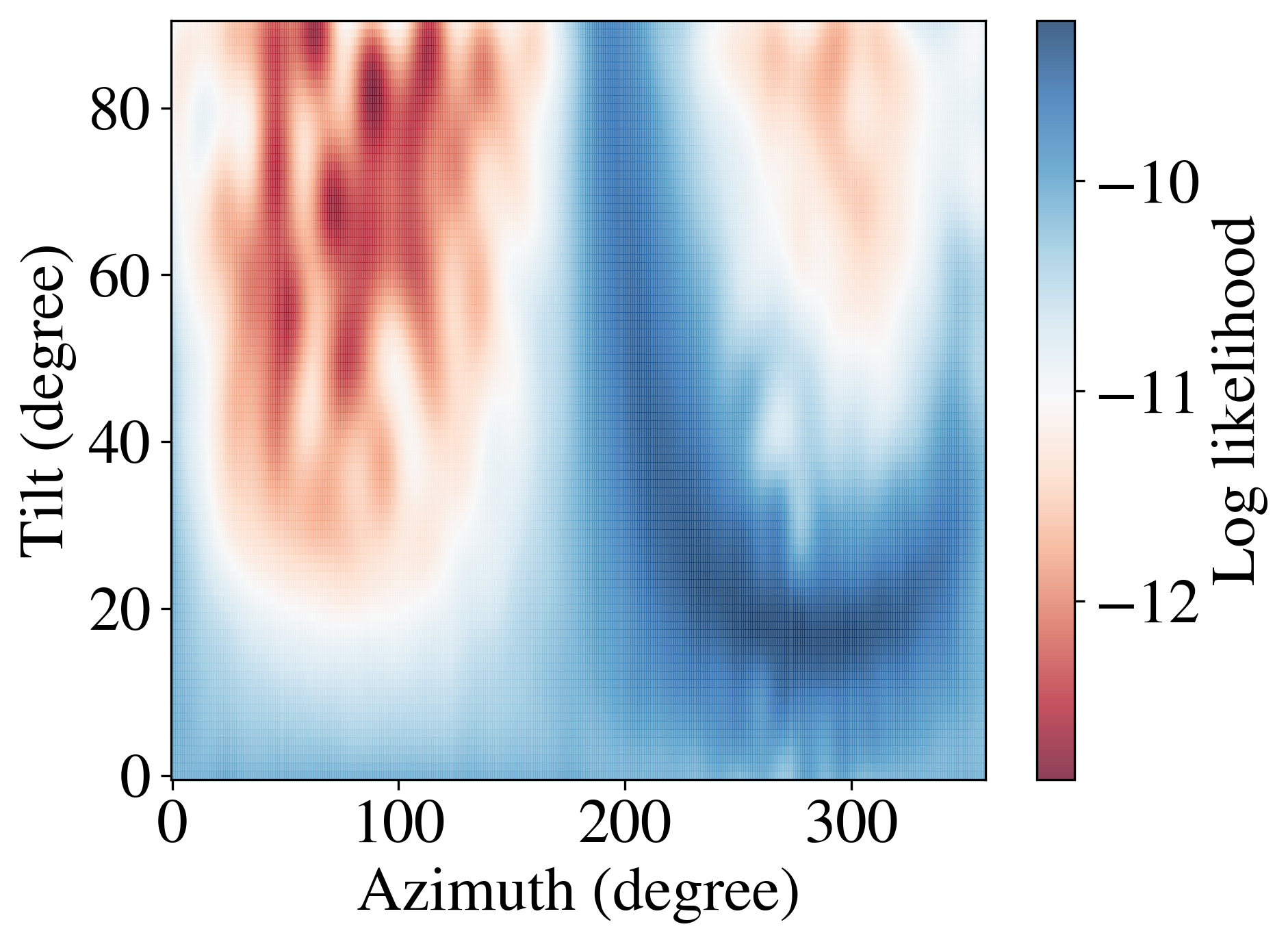}\label{fig:LLM}}\quad
    \subfloat[][]{
        \includegraphics[width=0.3\textwidth]{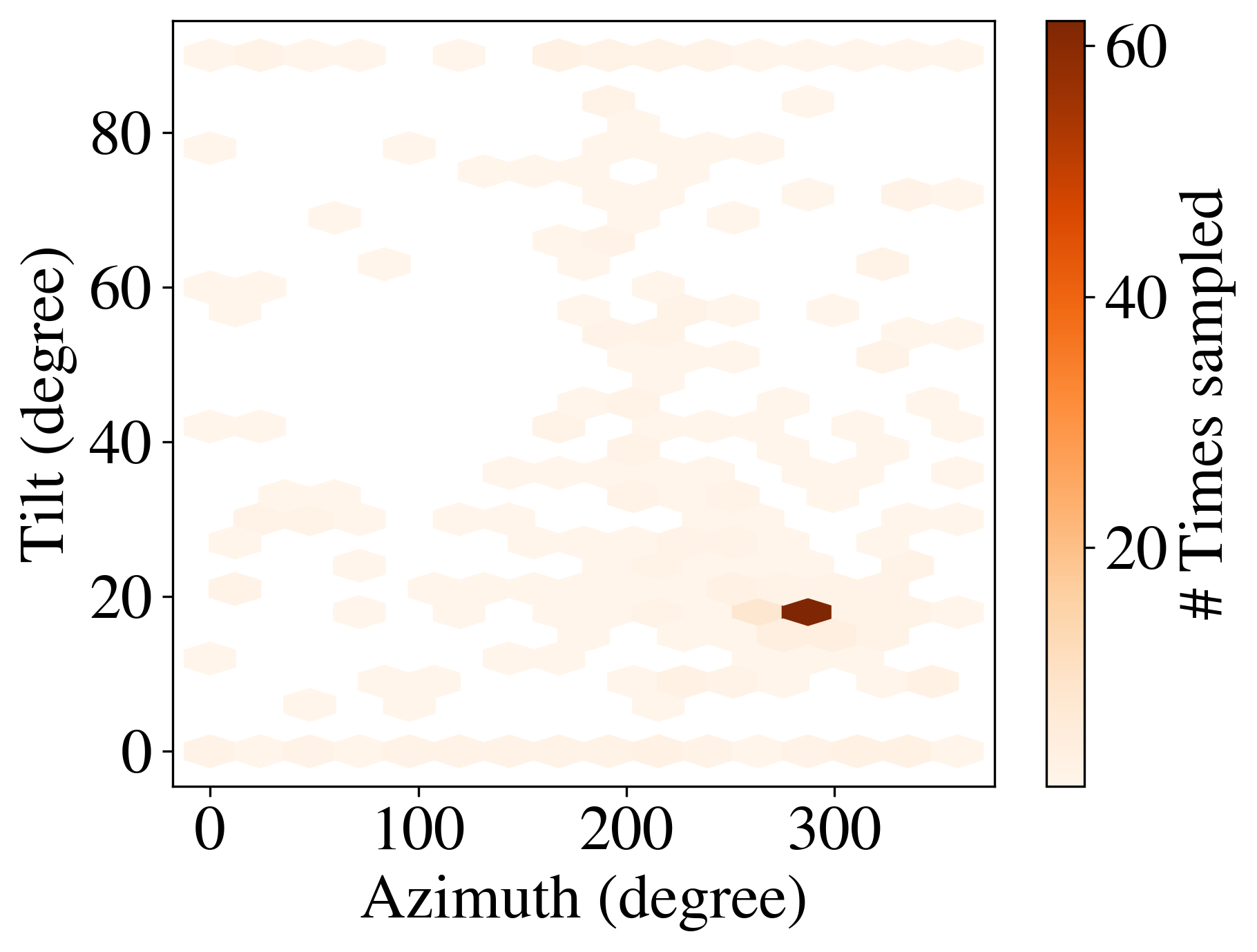}\label{fig:opt}}\quad
    \subfloat[][]{
        \includegraphics[width=0.3\textwidth]{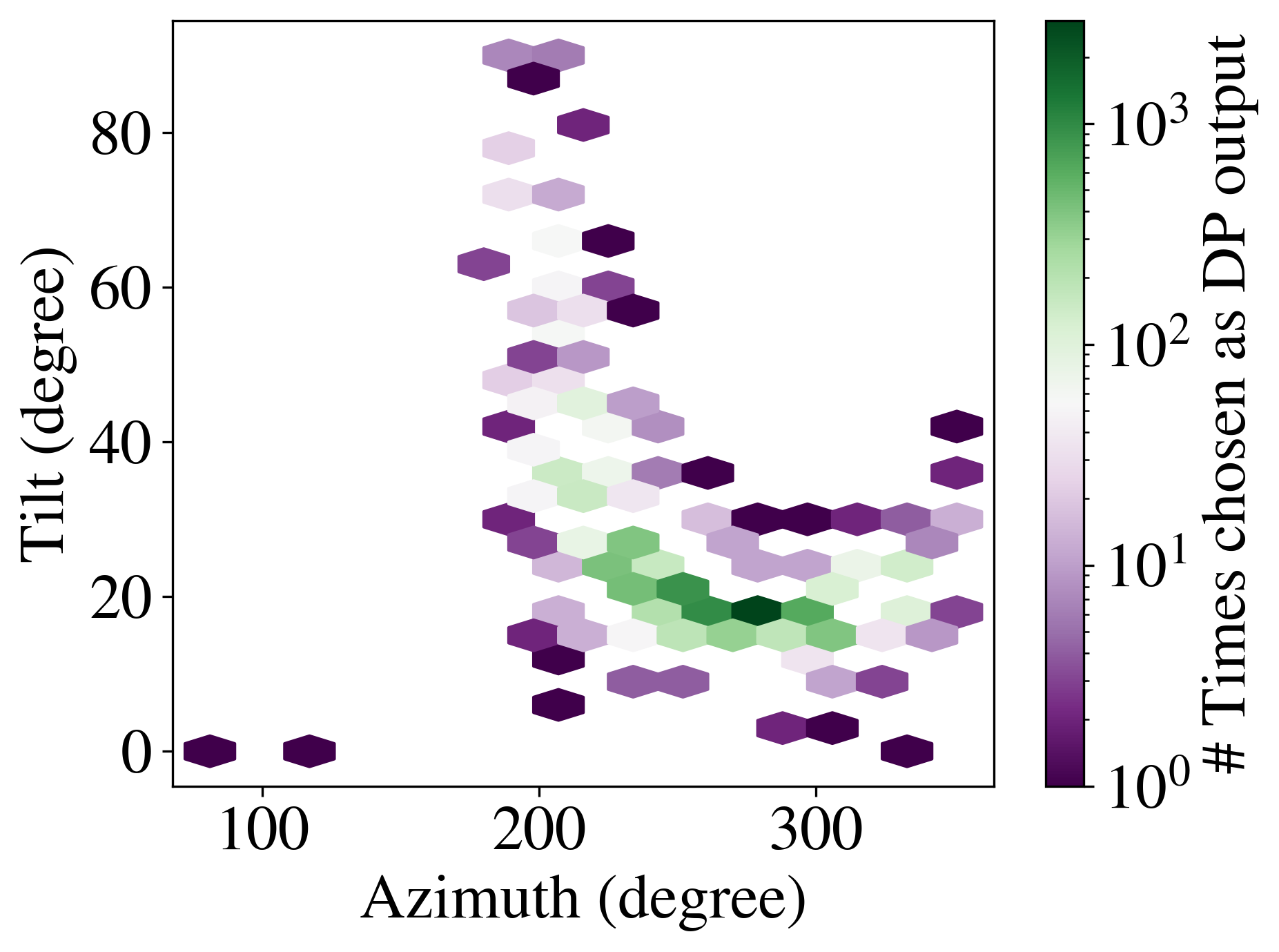}\label{fig:sample}}
    \caption{(\subref{fig:LLM}) Log likelihood of a particular $\blambda$ being the maximizer for $\epsilon=0.1$ and $\delta=0.1$.  (\subref{fig:opt}) The number of times a hexagonal region was sampled during the optimization stage. \subref{fig:sample}) The number of times (in $10,000$ draws) a hexagonal region was chosen to be the DP output according to Algorithm 1 are overlaid on top of the log-likelihood.}
    \label{fig:results}
\end{figure*}
\begin{figure}
    \centering
    \includegraphics{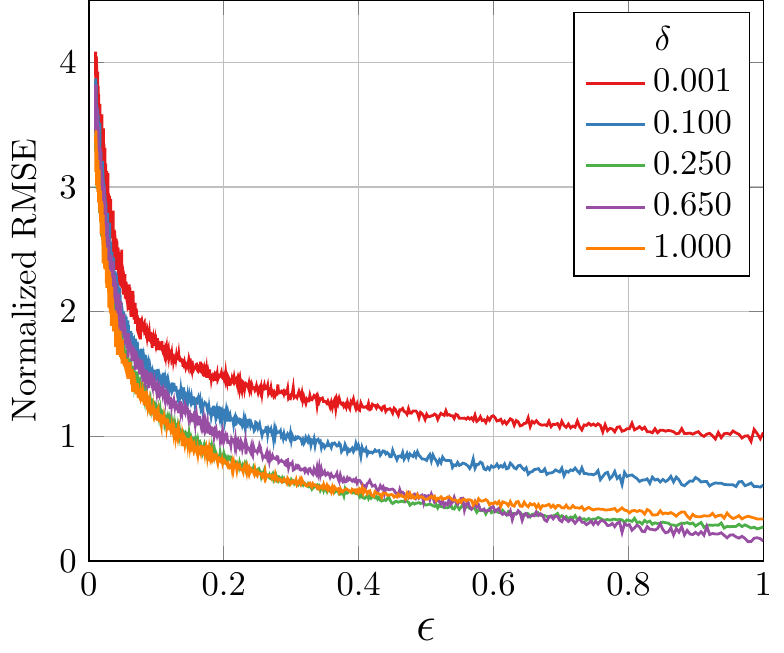}
    \caption{Normalized RMSE between the expected DP response and the optimal value of the metadata, $\sblambda$, for various values of $\epsilon$ and $\delta$. For each pair of $(\epsilon, \delta)$, $10,000$ samples were drawn using the mechanism to calculate the expected DP output.}
    \label{fig:dp}
\end{figure}
In this section, using hourly solar generation data from solar PV systems installed in California from the year $2016$, we numerically validate the DP-BO Solar PV metadata inference algorithm by showcasing the BO metadata inference step's accuracy and the DP mechanism's performance.

The performance validation of the solar metadata inference task is often impossible without the existence of labeled data. This dataset provided labeled surface azimuths for all the systems and surface tilts for some systems. Thus, the validation of the performance of the inference step is solely dependent on the error in the inferred surface azimuth. 

Some aspects of a solar PV system's data collection pipeline degrade the performance of the inference algorithm. For example, some systems in the dataset have solar panels from more than one brand and are often installed on different occasions. The ground truth of whether the various sub-systems all have the same configuration (tilt, azimuth, elevation, etc.) was not available. The systems with more than one known sub-system of PV panels report a single data stream of solar power generation, so it is difficult to infer the solar metadata information without a complicated disaggregation methodology in place. Finally, some installations contain a panel with the ability to track the sun over the course of a day. 

With these aspects in mind, we applied the DP-BO metadata inference mechanism to a solar PV system with a reported azimuth of $270\degree$ and tilt of $18\degree$, and its hourly solar generation profile from the year $2016$. 
The performance of the BO algorithm is measured in terms of its ability to hone in on the maximizer of the fit score function over the number of iterations of the algorithm. In \Cref{fig:LLM}, we show the posterior log-likelihood of azimuth and tilt pairs after $300$ iteration. It is evident that the algorithm places a high likelihood on the hexagonal segments in the neighborhood of $(270\degree, 18\degree)$. In \Cref{fig:opt}, we show the number of times in the $100$ samples, a hexagonal region was sampled. It is also evident that the GP-UCB acquisition function is performing well, as it has sampled the hexagonal segments around the true point the most. The point with the highest log likelihood (-0.31) after $T=100$ iteration is $\blambda=[$263\degree$, $18\degree$]$.

We next used the DP-BO mechanism $10,000$ on the same training run to output $10,000$ DP answers, $\tblambda$. After binning these output points into the aforementioned hexagonal segments, we plot the results in \Cref{fig:sample}. Again, it is evident that the DP-BO algorithm is able to consistently infer the metadata with high accuracy. \Cref{fig:dp}, which plots the normalized RMSE between the expected DP response and the optimal value of the metadata, $\sblambda$, for various values of $\epsilon$ and $\delta$, illustrates this point better. As expected, we can observe that the RMSE reduces as the privacy budget increases, and for a given $\epsilon$, the RMSE reduces as $\delta$ increases.

These experiments and the results provide a good validation for the proposed DP-BO algorithm to infer solar PV metadata.

\section{Conclusion}\label{sec:Conclusion}
In this paper, we discussed the growing solar PV infrastructure and its implication on grid operations and energy markets. To meet the needs of the grid operators, we proposed a methodology to infer solar PV metadata. Specifically, to perform the inference of solar PV metadata in a scalable, differentially private manner, we proposed the novel Bayesian optimization based DP algorithm that guarantees a desired level of $(\epsilon, \delta)$-DP. Additionally, in stark contrast to the industry standard 15/15 rule, the proposed methodology allows a utility to publish solar PV metadata to untrusted third-party analysts while protecting customer privacy with analytically provable privacy guarantees. Finally, we validated the proposed algorithm via numerical experiments. 

\bibliographystyle{ieeetr}
\bibliography{ref}

\begin{thebibliography}{10}

\bibitem{currie2023data}
R.~Currie, S.~Peisert, A.~Scaglione, A.~Shumavon, and N.~Ravi, ``Data
  {{Privacy}} for the {{Grid}}: {{Toward}} a {{Data Privacy Standard}} for
  {{Inverter-Based}} and {{Distributed Energy Resources}},'' {\em To appear in
  IEEE Power and Energy Magazine: Grid Communications}, 2023.

\bibitem{unitedstatesdepartmentofenergyofficeofelectricity2020strategy}
{United States Department of Energy Office of Electricity}, ``Modern
  distribution grid project: Strategy and {{Implementation Planning
  Guidebook}},'' Tech. Rep. Volume IV, {United States Department of Energy
  Office of Electricity}, 2020.

\bibitem{liu2012cyber}
J.~Liu, Y.~Xiao, S.~Li, W.~Liang, and C.~L.~P. Chen, ``Cyber {{Security}} and
  {{Privacy Issues}} in {{Smart Grids}},'' {\em IEEE Communications Surveys \&
  Tutorials}, vol.~14, no.~4, pp.~981--997, 2012.

\bibitem{lamm2021data}
T.~Lamm and E.~N. Elkind, ``Data for {{Grid Decarbonization}},'' tech. rep.,
  {Bank of America, UC Berkeley School of Law, UCLA School of Law}, Feb. 2021.

\bibitem{phuangpornpitak2016study}
N.~Phuangpornpitak and W.~Prommee, ``A {{Study}} of {{Load Demand Forecasting
  Models}} in {{Electric Power System Operation}} and {{Planning}},'' {\em
  GMSARN International Journal}, vol.~10, pp.~9--24, 2016.

\bibitem{frew2021curtailment}
B.~Frew, B.~Sergi, P.~Denholm, W.~Cole, N.~Gates, D.~Levie, and R.~Margolis,
  ``The curtailment paradox in the transition to high solar power systems,''
  {\em Joule}, vol.~5, pp.~1143--1167, May 2021.

\bibitem{meng2020datadriven}
B.~Meng, R.~C. G.~M. Loonen, and J.~L.~M. Hensen, ``Data-driven inference of
  unknown tilt and azimuth of distributed {{PV}} systems,'' {\em Solar Energy},
  vol.~211, pp.~418--432, Nov. 2020.

\bibitem{ferc2020new}
{Federal Energy Regulatory Commission}, ``A {{New Day}} for {{Distributed
  Energy Resources}},'' Tech. Rep. FERC Order No. 2222, {Federal Energy
  Regulatory Commission}, Sept. 2020.

\bibitem{americanpublicpowerassociation2018comments}
{American Public Power Association}, ``Comments on {{Docket No}}. {{RMI}}
  18-9-000: {{Participation}} of {{Distributed Energy Resource Aggregations}}
  in {{Markets Operated}} by {{Regional Transmission Organizations}} and
  {{Independent System Operators}},'' 2018.

\bibitem{bird2018review}
L.~A. Bird, F.~{Flores-Espino}, C.~M. Volpi, K.~B. Ardani, D.~Manning, and
  R.~McAllister, ``Review of {{Interconnection Practices}} and {{Costs}} in the
  {{Western States}},'' Tech. Rep. NREL/TP-6A20-71232, {National Renewable
  Energy Lab.(NREL), Golden, CO (United States); Western Interstate Energy
  Board (WIEB), Denver, CO (United States)}, Apr. 2018.

\bibitem{killinger2018search}
S.~Killinger, D.~Lingfors, Y.-M. {Saint-Drenan}, P.~Moraitis, W.~{van Sark},
  J.~Taylor, N.~A. Engerer, and J.~M. Bright, ``On the search for
  representative characteristics of {{PV}} systems: {{Data}} collection and
  analysis of {{PV}} system azimuth, tilt, capacity, yield and shading,'' {\em
  Solar Energy}, vol.~173, pp.~1087--1106, Oct. 2018.

\bibitem{montano-martinez2021detailed}
K.~{Montano-Martinez}, S.~Thakar, V.~Vittal, R.~Ayyanar, and C.~Rojas,
  ``Detailed {{Primary}} and {{Secondary Distribution System Feeder Modeling
  Based}} on {{AMI Data}},'' in {\em 2020 52nd {{North American Power
  Symposium}} ({{NAPS}})}, pp.~1--6, Apr. 2021.

\bibitem{kipp2018solar}
{Kipp \& Zonen}, ``Solar {{Irradiance Monitoring}} in {{Solar Energy
  Projects}}.''
  http://www.kippzonen.com/Download/810/Brochure-Solar-Irradiance-Monitoring-in-Solar-Energy-Projects,
  2018.

\bibitem{zipp2018when}
K.~Zipp, ``When is third-party solar monitoring needed?.''
  https://www.solarpowerworldonline.com/2018/05/when-is-third-party-solar-monitoring-needed/,
  May 2018.

\bibitem{dobos2014pvwatts}
A.~Dobos, ``{{PVWatts Version}} 5 {{Manual}},'' Tech. Rep. NREL/TP-6A20-62641,
  1158421, {National Renewable Energy Lab (NREL), Golden, CO (United States)},
  Sept. 2014.

\bibitem{sengupta2018national}
M.~Sengupta, Y.~Xie, A.~Lopez, A.~Habte, G.~Maclaurin, and J.~Shelby, ``The
  {{National Solar Radiation Data Base}} ({{NSRDB}}),'' {\em Renewable and
  Sustainable Energy Reviews}, vol.~89, pp.~51--60, June 2018.

\bibitem{rasmussen2005gaussianb}
C.~E. Rasmussen and C.~K.~I. Williams, {\em Gaussian {{Processes}} for
  {{Machine Learning}}}.
\newblock {The MIT Press}, Nov. 2005.

\bibitem{srinivas2012informationtheoretic}
N.~Srinivas, A.~Krause, S.~M. Kakade, and M.~W. Seeger,
  ``Information-{{Theoretic Regret Bounds}} for {{Gaussian Process
  Optimization}} in the {{Bandit Setting}},'' {\em IEEE Transactions on
  Information Theory}, vol.~58, pp.~3250--3265, May 2012.

\bibitem{ravi2022differentially}
N.~Ravi, A.~Scaglione, S.~Kadam, R.~Gentz, S.~Peisert, B.~Lunghino,
  E.~Levijarvi, and A.~Shumavon, ``Differentially {{Private K-Means Clustering
  Applied}} to {{Meter Data Analysis}} and {{Synthesis}},'' {\em IEEE
  Transactions on Smart Grid}, vol.~13, pp.~4801--4814, Nov. 2022.

\bibitem{cpuc2011proposed}
{Public Utility Commission of the State of Colorado}, ``Proposed {{Rules
  Relating}} to {{Smart Grid Data Privacy}} for {{Electric Utilities}},'' Tech.
  Rep. Decision No. R11-0922, {Public Utility Commission of the State of
  Colorado}, {Denver, Colorado}, 2011.

\bibitem{dwork2006calibrating}
C.~Dwork, F.~McSherry, K.~Nissim, and A.~Smith, ``Calibrating {{Noise}} to
  {{Sensitivity}} in {{Private Data Analysis}},'' in {\em Theory of
  {{Cryptography}}} (S.~Halevi and T.~Rabin, eds.), Lecture {{Notes}} in
  {{Computer Science}}, ({Berlin, Heidelberg}), pp.~265--284, {Springer}, 2006.

\bibitem{dwork2006our}
C.~Dwork, K.~Kenthapadi, F.~McSherry, I.~Mironov, and M.~Naor, ``Our {{Data}},
  {{Ourselves}}: {{Privacy Via Distributed Noise Generation}},'' in {\em
  Advances in {{Cryptology}} - {{EUROCRYPT}} 2006} (S.~Vaudenay, ed.), Lecture
  {{Notes}} in {{Computer Science}}, ({Berlin, Heidelberg}), pp.~486--503,
  {Springer}, 2006.

\bibitem{mcsherry2007mechanism}
F.~McSherry and K.~Talwar, ``Mechanism {{Design}} via {{Differential
  Privacy}},'' in {\em 48th {{Annual IEEE Symposium}} on {{Foundations}} of
  {{Computer Science}} ({{FOCS}}'07)}, pp.~94--103, Oct. 2007.

\end{thebibliography}

\appendices
\section{DP guarantees of DP-BO}\label[appendix]{app:proof}
To establish the DP guarantees of \Cref{alg:DP}, we need to first prove that the global sensitivity of the fit score function, $\loss$, is upper bounded with a high probability. This is stated and proved in the following lemma:
\begin{lemma}\label{lem:sensitivity_bound}
	The global sensitivity of the BO estimate of the fit score function is upper bounded with high probability, and this bound is given by:
	\begin{equation}
		|\mu_T'(\blambda) - \mu_T(\blambda)| \leq 2\left(\sqrt{\bar{\phi}_{T+1}} + \sqrt{\nu}\right),
	\end{equation}
	where $\bar{\phi}_{t} = 2\log(|\bLambda|t^2\pi^2/(2\delta))$ and $\nu = \log(6|\bLambda|/\delta)$.
\end{lemma}
\begin{proof}
	Consider the global sensitivity of the BO estimate of the fit score function at time $t$:
	\begin{equation}
		\Delta_{\mu_t} := \max_{\bm{p},\bm{p}' \subseteq \mathcal{X}, \blambda \in \bLambda} |\mu_t'(\blambda) - \mu_t(\blambda)|,
	\end{equation}
	where $\mu_t(\cdot)$ and $\mu_t'(\cdot)$ are the mean function given in \cref{eq:BO_mu_t} when the underlying solar generation profile is $\bm{p}$ and $\bm{p}'$, respectively. The term $|\mu_t'(\blambda) - \mu_t(\blambda)|$ may be written as follows:
	\begin{align}
		 & |\mu_t'(\blambda) - \mu_t(\blambda)|\nonumber \\
		 & = |\mu_t'(\blambda) - \ploss(\blambda) + \ploss(\blambda) -\loss(\blambda) + \loss(\blambda)- \mu_t(\blambda)|\nonumber   \\
		 & \leq |\mu_t'(\blambda)\!-\!\ploss(\blambda)|\! +\! |\ploss(\blambda)\! -\!\! \loss(\blambda)| \!+\! |\loss(\blambda)\! -\! \mu_t(\blambda)|.
	\end{align}
	We know from \Cref{lem:gp-ucb-bound} that
	\begin{equation}
		Pr\{ |\loss(\blambda) - \mu_t(\blambda)| \leq \sqrt{\phi_{t+1}} \sigma_t(\blambda)\} \geq 1-\delta,\label{eq:gp-ucb-bound1}
	\end{equation}
	where $\phi_{t} = 2\log(|\bLambda|t^2\pi^2/6\delta)$,  $0<\delta\lll 1$, and $\sigma_t(\cdot)$ is the standard deviation of the BO estimate of $\loss$ at time $t$.
	This bound similarly holds for $|\ploss(\blambda) - \mu_t'(\blambda)|$ and it is given by:
	\begin{equation}
		Pr\{ |\ploss(\blambda) - \mu_t'(\blambda)| \leq \sqrt{\phi_{t+1}} \sigma_t'(\blambda)\} \geq 1-\delta,\label{eq:gp-ucb-bound2}
	\end{equation}
	where $\sigma_t'(\cdot)$ is the standard deviation of the BO estimate of $\loss$ at time $t$ with the underlying generation profile being $\bm{p}'$.

	In a similar vain, consider $\ploss(\blambda) - \loss(\blambda)$. From our assumptions in \Cref{sec:BayesOpt}, $\loss, \ploss \sim \mathcal{N}(0, \Sigma(\blambda, \blambda)) \equiv \mathcal{N}(0,1)$. Thus, we have: 
	\begin{equation}
		\left(\ploss(\blambda) - \loss(\blambda)\right)/\sqrt{2} \sim \mathcal{N}(0, 1),
	\end{equation}
	and consequently, from the properties of sub-Gaussian random variables, we have the following inequality for $\blambda \in \bLambda$, $\forall c\in \mathbb{R}$:
	\begin{equation}
		Pr\left[\left|\ploss(\blambda) - \loss(\blambda)\right| \geq \sqrt{2}c \right] \leq  2\exp\left(-c^2/2\right).
	\end{equation} 
	This can in turn be used to deduce an upper bound using the union bound as follows:
	\begin{equation}
		Pr\left[\bigcup_{\lambda \!\in \!\bLambda} \!\left\{ \!\left|\ploss(\blambda) \!-\! \loss(\blambda)\right| \!\leq\! \sqrt{2}c \!\right\} \!\right] \!\!\geq\!\! 1\! -\! 2|\bLambda|\exp(\!-\!c^2\!/2)
	\end{equation}
	Setting $\delta:=2|\bLambda|\exp\left(-\frac{c^2}{2}\right)$, we get $c=\sqrt{2\log(2|\bLambda|/\delta)}$ and the following upper bound with high probability:
	\begin{equation}
		Pr\left[\bigcup_{\lambda \in \bLambda} \left\{ \left|\ploss(\blambda) - \loss(\blambda)\right| \leq \sqrt{2}c \right\} \right] \geq 1 - \delta\label{eq:DP-bound}
	\end{equation}

	Finally, we can now bound $\Delta_{\mu_t}$ as follows:
	\begin{equation}
		Pr\left[|\mu_t'(\blambda) - \mu_t(\blambda)| \leq 2\left(\sqrt{\bar{\phi}_{t+1}} + \sqrt{\nu}\right) \right] \geq 1-\delta.
	\end{equation}
	where $\bar{\phi}_{t+1} := 2\log(|\bLambda|t^2\pi^2/2\delta)$ and $\nu := \log(6|\bLambda|/\delta)$. This inequality is obtained via the following steps. First, we apply the following transformations to \crefrange{eq:gp-ucb-bound1}{eq:gp-ucb-bound2} and \cref{eq:DP-bound}:
	\begin{subequations}\label{eq:bounds_simplified}
		\begin{align}
			&Pr\left[ |\loss(\blambda) - \mu_t(\blambda)| \leq \sqrt{\bar{\phi}_{t+1}}\right] \geq 1-\delta/3,\label{eq:gp-ucb-bound3}\\
			&Pr\left[ |\ploss(\blambda) - \mu_t'(\blambda)| \leq \sqrt{\bar{\phi}_{t+1}}\right] \geq 1-\delta/3,\label{eq:gp-ucb-bound4}\\
			&Pr\left[\bigcup_{\lambda \in \bLambda} \left\{ \left|\ploss(\blambda) - \loss(\blambda)\right| \leq 2\sqrt{\nu} \right\} \right] \geq 1 - \delta/3\label{eq:DP-bound2},
		\end{align}
	\end{subequations}
	where $\bar{\phi}_{t+1}:= 2\log(|\bLambda|t^2\pi^2/2\delta)$, $\nu:=\log(6|\bLambda|/\delta)$, and we used the assumption from \Cref{sec:BayesOpt_GP} -- that the kernel $\Sigma$ is bounded above by $1$ -- to bound $\sigma_t(\blambda)$ and $\sigma'_t(\blambda)$ from above by $1$.
	We next define the following events:
	\begin{subequations}\label{eq:events}
		\begin{align}
			E_1: |\loss(\blambda) - \mu_t(\blambda)| &\leq \sqrt{\bar{\phi}_{t+1}} \\
			E_2: |\ploss(\blambda) - \mu_t'(\blambda)| &\leq \sqrt{\bar{\phi}_{t+1}} \\
			E_3: \left|\ploss(\blambda) - \loss(\blambda)\right| &\leq 2\sqrt{\nu}\\
			E: |\mu_t'(\blambda) - \mu_t(\blambda)| &\leq 2\left(\sqrt{\bar{\phi}_{t+1}} + \sqrt{\nu}\right)
		\end{align}
	\end{subequations}
	Events $E_1, E_2$ and $E_3$ all satisfy the following bound:
	\begin{equation}
		Pr[E_i] \geq 1 - \delta/3, \quad \text{for } i \in \{1,2,3\}.
	\end{equation}
	By definition, it is clear that $\bigcap_{i=1}^3 E_i \subseteq E$. Thus:
	\begin{align}
		Pr[E] & \geq Pr\left[\bigcap_{i=1}^3 E_i\right] = 1 - Pr\left[\bigcup_{i=1}^3 E_i^c\right] \geq 1 - \delta.\label{eq:Event-E}
	\end{align}
\end{proof}
It is now left to prove that if the above lemma holds, then \Cref{alg:DP} is $(\epsilon, \delta)$-DP.
\begin{lemma}[\Cref{alg:DP} is $(\epsilon, \delta)$-DP]\label{lem:exponential_mechanism_proof}
	Let $\bblambda_{\ell}$ denote the mechanism presented in \Cref{alg:DP}. If \Cref{lem:sensitivity_bound} holds, then the outcome, $\tblambda$, of $\bblambda_{\ell}$ is $(\epsilon, \delta)$-DP for any pair of neighboring datasets $\bm{p}$ and $\bm{p}' \in \mathcal{X}$.
\end{lemma}
\begin{proof}
	If $\bm{p}, \bm{p}'$ are two neighboring datasets in $\mathcal{X}$ and the sensitivity bound in \Cref{lem:sensitivity_bound} holds, then choosing $\tblambda$ with probability proportional to $\exp\left\{\frac{\epsilon\mu_T(\blambda)}{4\left(\sqrt{\bar{\phi}_{T+1}} + \sqrt{\nu}\right)}\right\}$ is $\epsilon$-DP from \Cref{thm:exp_mechanism_is_eDP}, i.e.,
	\begin{equation}
		Pr[\bblambda_{\loss} = \tblambda\mid E] \leq e^\epsilon Pr[\bblambda_{\ploss} = \tblambda\mid E].\label{eq:exponential_eps}
	\end{equation}
	Furthermore, considering $Pr[\bblambda_{\loss} = \tblambda]$, we have:
	\begin{align}
		 &Pr[\bblambda_{\loss} = \tblambda] \nonumber\\
      &= Pr[\bblambda_{\loss}=\tblambda\mid E] Pr[E] + Pr[\bblambda_{\loss}=\tblambda\mid\overline{E}] (1 - Pr[E])\nonumber\\
	   &\leq  Pr[\bblambda_{\loss}=\tblambda\mid E] Pr[E] + 1 \times \delta\nonumber\\
		 &\leq e^\epsilon Pr[\bblambda_{\ploss}=\tblambda\mid E] Pr[E] + \delta  \nonumber\\
      &\leq e^\epsilon Pr[\bblambda_{\ploss}=\tblambda] + \delta\nonumber,
	\end{align}
	where the first and the second inequalities follow from \cref{eq:Event-E} and \cref{eq:exponential_eps}, respectively.
\end{proof}%
This completes the proof of \Cref{thm:exp_mechanism_is_eDP}.
\end{document}